\documentclass[lettersize,journal]{IEEEtran}
\usepackage{amsmath,amsfonts}
\usepackage{algorithmic}
\usepackage{algorithm}
\usepackage{array}
\usepackage[caption=false,font=normalsize,labelfont=sf,textfont=sf]{subfig}
\usepackage{textcomp}
\usepackage{stfloats}
\usepackage{url}
\usepackage{verbatim}
\usepackage{graphicx}
\usepackage{cite}
\usepackage{amsthm}
\usepackage{latexsym,amssymb,verbatim,enumerate,graphicx}
\usepackage{hyperref}
\usepackage{color}
\usepackage[colorinlistoftodos]{todonotes}
\usepackage{MnSymbol}

\newtheorem{thm}{Theorem}
\newtheorem*{thm*}{Theorem}
\newtheorem{definition}{Definition}

\newtheorem*{prop*}{Proposition}
\newtheorem{lem}{Lemma}
\newtheorem*{lem*}{Lemma}
\newtheorem{cor}{Corollary}
\newtheorem*{cor*}{Corollary}

\newtheorem{fact}{Fact}

\definecolor{KKgreen}{RGB}{0,100,0}

\newcommand{\ot}{\otimes}

\def\supp{{\rm supp}}

\def\cA{{\mathcal A}}
\def\cB{{\mathcal B}}

\def\cE{{\mathcal E}}
\def\cF{{\mathcal F}}
\def\cG{{\mathcal G}}
\def\cH{{\mathcal H}}
\def\cI{{\mathcal I}}
\def\cM{{\mathcal M}}
\def\cN{{\mathcal N}}
\def\cK{{\mathcal K}}

\def\cR{{\mathcal R}}
\def\cS{{\mathcal S}}
\def\cT{{\mathcal T}}
\def\cU{{\mathcal U}}
\def\cV{{\mathcal V}}

\def\cHil{{\cH_{B_i^L}}}
\def\cHir{{\cH_{B_i^R}}}

\def\tr{{\rm tr}}
\def\id{{\rm id}}
\def\rank{{\rm rank}}
\def\rag{\rangle}
\def\lag{\langle}
\def\rrag{\rrangle}
\def\llag{\llangle}

\def\half{\frac{1}{2}}
\def\tA{{\tilde A}}
\def\tB{{\tilde B}}
\def\supp{{\rm supp}}
\def\alg{{\rm Alg}}
\def\dim{{\rm dim}}
\def\argmin{{\rm argmin}}
\def\Sch{{\rm SchR}}
\begin{document}

\title{Exact and local compression of quantum bipartite states}

\author{Kohtaro Kato,~\IEEEmembership{Department of Mathematical Informatics, \\
Graduate School of Informatics, \\
Nagoya University, Nagoya 464-0814, Japan}}



\maketitle
\begin{abstract}
We study exact local compression of a quantum bipartite state; that is, applying local quantum operations to reduce the dimensions of the Hilbert spaces while perfectly preserving the correlation. We provide a closed-form expression for the minimal achievable dimensions, formulated as a minimization of the Schmidt rank of a particular pure state constructed from the given state. Numerically tractable upper and lower bounds on this rank are also obtained. As an application, we consider the exact compression of quantum channels. This method enables the analysis of a post-processing step that reduces the output dimensions while retaining the information contained in the original channel’s output.
\end{abstract}


\section{Introduction}
\IEEEPARstart{Q}{uantum} data compression is one of the most fundamental quantum information processing. Its concept is analogous to Shannon's classical data compression and aims to reduce the dimensions of the storage of quantum states while minimizing information loss. 
Various approaches to quantum compression have been proposed, primarily focusing on achieving approximately or asymptotically accurate representations of quantum states~\cite{Schumacher95,Rahul05, Datta13, Khanian19,Khanian2019-bg}. These protocols have yielded valuable insights into the trade-offs between compression efficiency, fidelity, and additional resources.

Nevertheless, there remains room to investigate {\it exact} compression of quantum states, where the compressed representation fully retains the information content of the original state. 
This compression type has gained interest in information theory and the condition is known as the sufficiency of  statistics~\cite{Fisher22,Coverthomas06}. 
Classical statistics are {\it sufficient} for a given statistical model if they are as informative as the original model. 
It is well known that the minimal random variable for describing a statistical model (a family of probability distributions) is given as the minimal sufficient statistics associated with it~\cite{LehmannScheffe50}.
The concept of sufficiency has been extended to quantum systems~\cite{Petz1986-nm,Petz1988-ht}, where the concept of sufficient statistics is replaced by that of sufficient subalgebras.

In this work, we study the local and exact compressions of bipartite quantum states. In other words, we consider quantum operations in each subsystem (local operations) with a smaller output Hilbert space that can be reverted by another quantum operation. This type of data compression is an exact and noiseless one-shot quantum data compression of general mixed state sources without side information or entanglement assistance. Here, it is crucial to consider mixed state source, since pure state source cannot be compressed less than the Schmidt rank. 

An asymptotic scenario of local compressions of mixed state sources has been studied in Refs.~\cite{PhysRevA.57.3364,HowardBarnum_2001,PhysRevLett.87.017902,9624960} for classical-quantum states and extended for fully quantum states in Ref.~\cite{Khanian2019-bg}. The asymptotically optimal rate in Ref.~\cite{Khanian2019-bg} is given by the entropy of the state restricted on the subalgebra defined via the Koashi-Imoto decomposition~\cite{Koashi2002-vv,Hayden2004-os}. Similarly, one can check that the minimal dimension of (one-shot) exact local compression is also given by a subspace defined via the Koashi-Imoto decomposition. However, the explicit calculation of the Koashi-Imoto decomposition is highly complicated~\cite{Koashi2002-vv,Yamasaki2019-pf}, and thus no closed formula for the optimal rate has been obtained so far.  

As a result, we show a closed formula to calculate the minimal dimension of the output Hilbert spaces. The formula is obtained by minimizing the Schmidt rank (i.e., the rank of the reduced matrix) over unitarily related pure states, calculated from the input state. As a corollary, we provide additional tractable lower and upper bounds for the minimal dimensions. Our result is based on the recent development of quantum sufficiency~\cite{Jencova2012-sy} and operator algebra quantum error correction~\cite{Beny2007-bw,Beny2007-ua,Chen2020-in}, which formulates the set of operators preserved by a given quantum channel.

The remainder of this paper is structured as follows. Section~\ref{sec:2} provides relevant background concepts. We provide an informal sketch of our strategy for the proof and summary of the main results in Section~\ref{sec:MR}. Section~\ref{sec:3} provides a new characterization of the minimal sufficient subalgebra in our setting, which plays an important role in validation. Section~\ref{sec:calculate} discusses the validation of the main results. Section~\ref{sec:example} presents applications of the main results. Finally, Section~\ref{sec:conclusion} concludes the paper by summarizing our findings and outlining potential directions for future research.

\section{Exact local compression of bipartite states}\label{sec:2}\subsection{Notations}
Throughout this study, we consider finite-dimensional Hilbert spaces. For a set of matrices $S$, $\alg\{S\}$ denotes the $C^*$-algebra generated by $S$. 
$\cB(\cH)$ denotes the set of bounded operators (the matrix algebra) on Hilbert space $\cH$. Lower indices of density matrix represents the corresponding Hilbert spaces. For instance, $\rho_{AB}$ is a density operator on $\cH_A\ot \cH_B$ and $\rho_A:=\tr_{B}\rho_{AB}$ is the reduced density matrix on $\cH_A$. $I_A$ is the identity operator on system $A$ and $\tau_A=I_A/d_A$ is the completely mixed state on $A$, where $d_A=\dim\cH_A$. We define $|\cI\rrag_{X{\bar X}}:=\sum_{k=1}^{d_X}|\phi^k\rag_{X}|\phi^k\rag_{\bar X}$ in an orthogonal basis $\{|\phi^k\rag\}_k$ of $\cH_{X}\cong\cH_{\bar X}$. Similarly, the normalized maximally entangled state is denoted as $|\Psi\rrag_{X{\bar X}}:=|\cI\rrag_{X{\bar X}}/\sqrt{d_X}$. For $O\in\cB(\cH_A)$, $O^{T_A}$ is the transposition, and ${\bar O}$ is the complex conjugate with a fixed basis of $A$. 
We will omit the tensor products with identities that are clear from the context. 

For a completely positive and trace-preserving (CPTP) map (i.e., a quantum channel) $\cE:\cB(\cH)\to\cB(\cK)$, $\cE^\dagger$ denotes its Hilbert-Schmidt dual 
\begin{equation*}
 \tr\left(O\cE(\rho) \right)=\tr\left(\cE^\dagger(O)\rho\right)\,\quad\forall \rho\in\cB(\cH), \forall O\in\cB(\cK)\,.
\end{equation*}
When $\cK=\cH$, $\cF(\cE^\dagger)$ denotes the set of the fixed-points of $\cE^\dagger$
\begin{equation*}
    \cF(\cE^\dagger):=\left\{O\in\cB(\cH) \:\middle|\: \cE^\dagger(O)=O\right\}\,,
\end{equation*}
which forms a unital $C^*$-algebra~\cite{choi1974schwarz}. 
$\cE^c$ denotes a complementary channel of $\cE$. Let $\{E_a\}$ be the Kraus operators of $\cE$. We then define the {\it correctable algebra} $\cA_\cE$ of $\cE$~\cite{Beny2007-bw,Beny2007-ua} on $\cH$ as follows: 
\begin{align*}
\cA_\cE&:=(\{E_a^\dagger E_b\}_{a,b})'\\
&\:=\left\{ X\in \cB(\cH)\;|\;[E_a^\dagger E_b,X]=0\;\forall a,b \right\}\,.
\end{align*}
This is the unital $C^*$-subalgebra of $\cB(\cH)$ that describes perfectly preserved information under the action of $\cE$ (see Appendix~\ref{app:OAQEC} for more details). 

For a state $\rho>0$, we denote the modular transformation as
\begin{equation*}
\Delta_\rho^t(\cdot):=\rho^{it}\cdot \rho^{-it}\,,\quad t\in\mathbb{R}  
\end{equation*}
and $\Gamma_\rho(\cdot):=\rho^{\half}\cdot\rho^{\half}$. We also denote the adjoint action of the generator of $\Delta_\rho^t$ as 
\begin{equation*}
    {\rm ad}_\rho(\cdot):=[\cdot,\log\rho]\,.
\end{equation*}
For a given CPTP map $\cE$, the Petz recovery map~\cite{Petz1986-nm,Petz1988-ht} with respect to $\rho>0$ is defined as follows: 
\begin{align}
    \cR^{\rho,\cE}(\cdot)&:=\rho^{\half}\cE^\dagger\left(\cE(\rho)^{-\half}\cdot\cE(\rho)^{-\half}\right)\rho^{\half}\,\label{eq:Petzmap}\\
    &=\Gamma_{\rho}\circ\cE^\dagger\circ\Gamma_{\cE(\rho)}^{-1}\nonumber\,.
\end{align}

For a bipartite pure state $\psi_{AB}=|\psi\rag\lag\psi|_{AB}$, $S(A)_\psi=-\tr\psi_A\log\psi_A$ is the entanglement entropy and $\Sch(A:B)_\psi$ is the Schmidt rank of $|\psi\rag_{AB}$.  $\Sch(A:B)_\psi$ is equal to the rank of the reduced state, i.e.,  $\Sch(A:B)_\psi=\rank(\psi_A)=\rank(\psi_B)$.

Lastly, for a given $\rho_{AB}$ with $\rho_B>0$, we define 
\begin{equation}\label{def:Choi1}
    J_{AB}:=\rho_B^{-\half}\rho_{AB}\rho_B^{-\half}
\end{equation}
By definition, $J_B=I_B$ and therefore this positive semi-definite operator can be regarded as the Choi-Jamiołkowski representation of a unital CP-map
\begin{equation}\label{def:Omega1}
    \Omega_{A\to B}^\dagger(X_A):=\tr_A\left(J_{AB}(X_A^{T_A}\otimes I_B)\right)\,.
\end{equation}
The dual is a CPTP-map $\Omega_{B\to A}$. 

\subsection{Exact local compression of bipartite states}
We are interested in the quantum bipartite state $\rho_{AB}$ in the finite-dimensional Hilbert space $\cH_{AB}=\cH_A\otimes \cH_B$. 
We assume without loss of generality that $\rho_A,\rho_B>0$ (by restricting each Hilbert space to the support of the reduced state). 
\begin{definition} We say a CPTP-map $\cE_{B\to\tB}:\cB(\cH_B)\to\cB(\cH_\tB)$ is an exact local compression of $\rho_{AB}$ on $B$, if there exists a CPTP-map $\cR_{\tB\to B}$ satisfying 
\begin{equation}\label{eq:recovery0}
    \cR_{\tB\to B}\circ\cE_{B\to \tB}(\rho_{AB})=\rho_{AB}.
\end{equation}
We say that compression $\cE_{B\to\tB}$ is nontrivial if $d_\tB<d_B$. We define the exact compression $\cE_{A\to\tA}$ on $A$ similarly.
\end{definition}
This study aims to calculate the minimal dimensions $d_{\tA}, d_{\tB}$ of the exact local compressions. It is sufficient to consider only the compressions on $B$ (or $A$) because of the symmetry of the problem.

\subsection{Koashi-Imoto decomposition}\label{sec:subKI-dec}
In Ref.~\cite{Koashi2002-vv}, Koashi and Imoto analyzed the structure of quantum operations that remain a set of classically labeled quantum states $\{\rho^x_B\}_{x\in\chi}$ unchanged. They found that for any set $\{ \rho_B^x\}$, there is a unique decomposition of $\cH_B$ that completely characterizes CPTP-maps preserving $\{\rho_B^x\}$. This idea was generalized to a fully quantum setup in Ref.~\cite{Hayden2004-os}. It has been shown that for any (finite-dimensional) quantum bipartite state $\rho_{AB}$, there is a unique decomposition called {\it the  Koashi-Imoto decomposition}. 
\begin{definition}
Consider a direct sum decomposition 
\begin{align}
    \cH_B&\cong\bigoplus_i\cHil\otimes\cHir\label{eq:KI-decomposition0}\\
    \rho_{AB}&=\bigoplus_i p_i \rho_{AB_i^L}\otimes \omega_{B_i^R}\,,\label{eq:KI-decomposition}
\end{align}
where $\{p_i\}$ is a probability distribution and $\rho_{AB_i^L}$ and $\omega_{B_i^R}$ are the states on $\cH_A\otimes\cHil$ and $\cHir$, respectively.
    A decomposition in Eqs.~\eqref{eq:KI-decomposition0}-\eqref{eq:KI-decomposition} is said to be the {\it Koashi-Imoto decomposition} if for any CPTP-map $\cE_B$ satisfying
    \begin{equation*}
        \cE_B(\rho_{AB})=\rho_{AB}\,,
    \end{equation*}
    any $\cE_B$’s Stinespring dilation isometry $V_{B\to BE}$ defined by
    \begin{equation*}
        \cE_B(\cdot)=\tr_E\left(V_{B\to BE}\cdot (V_{B\to BE})^\dagger \right)
    \end{equation*}
    is decomposed into 
    \begin{equation*}
        V_{B\to BE}=\bigoplus_iI_{B_i^L} \otimes V_{B_i^R\to B_i^R E}
    \end{equation*}
    satisfying 
    \begin{equation*}
        \tr_E\left(V_{B_i^R\to B_i^R E}\:\omega_{B_i^R}(V_{B_i^R\to B_i^R E})^\dagger\right)=\omega_{B_i^R}\,,\quad\forall i.
    \end{equation*}
\end{definition}
The Koashi-Imoto decomposition can be considered as a special case of the factorization theorem in Refs.~\cite{Mosonyi2004-wo, Jencova2006-on}, applicable for a general set of states, which is a quantum analogue of the classical Fisher–Neyman factorization theorem~\cite{Fisher22,Neyman1936-ht,Halmos1949-qb}. 

The factorization theorem~\cite{Mosonyi2004-wo, Jencova2006-on} further demonstrates that, if Eq.~\eqref{eq:recovery0} holds, then there exists a decomposition $\cH_\tB\cong\bigoplus_i\cH_{\tB_i^L}\otimes\cH_{\tB_i^R}$  such that $\cE_{B\to \tB}$ in Eq.~\eqref{eq:recovery0}, restricted on subspace $\cH_{\tB_i^L}\otimes\cH_{\tB_i^R}$, must have the form
\begin{equation*}
    \left.\cE_{B\to \tB}\:\right\vert_{\tB_i^L\tB_i^R}=\cU_{B_i^L\to \tB_i^L}\otimes \beta_{B_i^R\to\tB_iR }\,,
\end{equation*}
where $\cU_{B_i^L\to \tB_i^L}$ is a unitary map (therefore, $\cH_{B_i^L}\cong \cH_{\tB_i^L}$) and $\beta_{B_i^R\to\tB_iR }$ is a CPTP map. This implies that minimal compression is achieved when 
\begin{align*}
    \cH_{\tB}&\cong\bigoplus_i\cHil\\
    \rho_{A\tB}&=\bigoplus_i p_i \rho_{AB_i^L}\,.
\end{align*}
Note that precisely speaking, $\cH_{B_i^L}$ in the Koashi-Imoto decomposition are not necessarily orthogonal each other. However, one can always consider $\bigoplus_i\cH_{B_i^L}\ot|i\rangle\langle i|\cong\bigoplus_i\cH_{B_i^L}$ to orthogonalize them.

We summarize this subsection as follows.
\begin{fact}
    For a given $\rho_{AB}$, consider the Koashi-Imoto decomposition~\eqref{eq:KI-decomposition}. The minimal dimension of the exact compression for $\rho_{AB}$ is given by 
    \begin{equation}\label{eq:mindim}
        d_\tB^{\min}:=\sum_id_{B_i^L}\,.
    \end{equation}
\end{fact}

\subsection{Minimal sufficient subalgebra}
Here we introduce another characterization of the Koashi-Imoto decomposition based on the theory of quantum sufficiency. Consider a set of density matrices on $B$
\begin{equation}\label{def:mu}
    \cS:=\left\{\mu_B=\frac{\tr_A\left(M_A\rho_{AB}\right)}{\tr\left(M_A\rho_{A}\right)} \;|\; 0\leq M_A\leq I_A\right\}\,,
\end{equation}
dominated by $\rho_B$ (i.e., $\rho_B\in\cS$ and $\forall\mu_B\in\cS$, $\supp(\mu_B)\subset\supp(\rho_B)$). In terms of $\cS$, Eq.~\eqref{eq:recovery0} is rephrased as~\cite{Hayden2004-os,mixedphase1} 
\begin{equation}\label{eq:recovery1}
    \cR_{\tB\to B}\circ\cE_{B\to \tB}(\mu_{B})=\mu_{B}\,,\forall \mu_B\in\cS.
\end{equation}
This kind of condition has been extensively studied in the theory of sufficiency~\cite{Petz1986-nm,Petz1988-ht}, and it is known that the map $\cR_{\tB\to B}$ can always be selected as the Petz recovery map $\cR^{\rho,\cE}_{\tB\to B}$ with respect to $\rho_B$.

The theory of sufficiency states that the minimal dimension of $\tB$ is associated with the {\it minimal sufficient subalgebra} $\cM_B^S\subset\cB(\cH_B)$ of $\cS$. Because $\cM^S_B$ is a finite-dimensional $C^*$-algebra, a decomposition 
\begin{equation}\label{eq:decHBminimal}
 \cH_B\cong\bigoplus_i\cHil\otimes\cHir   
\end{equation}
exists such that
\begin{align}
\cM_B^S\cong &\bigoplus_i \cB(\cHil)\ot I_{B_i^R}\,,\label{eq:minimaldecoH}\\
(\cM_B^S)'\cong &\bigoplus_i I_{B_i^L}\otimes \cB(\cHir)\,,\nonumber
\end{align}
where $(\cM_B^S)':=\{X_B\in\cB(\cH_B)| [X_B,Y_B]=0\,,\; \forall Y_B\in \cM_B^S\}$ is the commutant of $\cM_B^S$. 
The sufficiency theorem~\cite{Petz1988-ht} states that Eq.~\eqref{eq:recovery1} implies $\cH_\tB$ must support algebra isomorphic to $\cM^S_B$. 
Therefore, the minimal Hilbert space $\cH_\tB$ must be:
\begin{equation*}
    \cH_\tB\cong \bigoplus_i\cHil\,.
\end{equation*}

As discussed in Sec.~\ref{sec:subKI-dec}, decomposition~\eqref{eq:decHBminimal} must match that of the Koashi-Imoto decomposition~\eqref{eq:KI-decomposition0}: an algorithm can be used to calculate the Koshi-Imoto decomposition~\cite{Koashi2002-vv,Yamasaki2019-pf}; however, these algorithms do not have an explicit formula to calculate the decomposition or the dimension $d_\tB$. For a similar reason, calculating the optimal compression rate in the asymptotic results in Ref.~\cite{Khanian2019-bg} is a non-trivial task.

\section{Main result}\label{sec:MR}
\subsection{Sketch of the proof}
In this section, we provide a sketch of the strategy to obtain our main result. The explicit proof is given in Sec.~\ref{sec:3} and Sec.~\ref{sec:calculate}, and readers who are interested in the statement of the main theorem can skip this subsection. 

Our main goal is to obtain the optimal dimension~\eqref{eq:mindim} from a given bipartite density matrix $\rho_{AB}$, without explicitly calculating the Koashi-Imoto decomposition. The main proof consists of two steps. First, we characterize $(\cM_B^S)'$ as the largest subalgebra of the fixed-point algebra of a certain CPTP-map $\cT_B$ that is invariant under modular transformation $\Delta_{\rho_B}^t(\cdot)=\rho_B^{it}\cdot \rho_B^{-it}$ (Sec.~\ref{sec:3}).

To obtain the characterization, notice that the Koashi-Imoto decomposition~\eqref{eq:KI-decomposition} implies that 
\begin{equation*}
    J_{AB}=\bigoplus_i \rho^{-\half}_{B_i^L}\rho_{AB_i^L}\rho^{-\half}_{B_i^L}\otimes I_{B_i^R}\,.
\end{equation*}
This is the Choi-Jamiołkowski representation of a unital CP-map $\Omega_{A\to B}^\dagger$. A recent result in quantum sufficiency~\cite{Jencova2012-sy} shows that the minimal sufficient subalgebra $\cM_B^S$ is generated by $\Delta_{\rho_B}^{t}\circ\Omega^\dagger_{A\to B}(M_A)$ for $M_A\in\cB(\cH_A), t\in\mathbb{R}$.
Although it is still hard to specify the explicit algebra structure of $\cM_B^S$, one can characterize the commutant $(\cM_B^S)'$ as the set of operators commuting with these generators. 
Since the generators of $\cM_B^S$ form the images of unital CP-maps $\Delta^t_{\rho_B}\circ\Omega^\dagger_{A\to B}$, a standard theorem in operator-algebra quantum error correction~\cite{Beny2007-ua} on the commutant of a unital CP-map provides an interpretation of $(\cM_B^S)'$ as the intersection of the correctable algebras of CPTP-maps parameterized by $t\in\mathbb{R}$. Further application of another result in error correction~\cite{Chen2020-in} implies that $(\cM_B^S)'$ is the largest subalgebra of the fixed points of $\cT_B:=\cR^{\tau,\Omega_{E\to B}}_{E\to B}\circ\Omega_{B\to E}$, where $\Omega_{B\to E}:=(\Omega_{B\to A})^c$, that is invariant under $\Delta_{\rho_B}^t$.

In the second (Sec.~\ref{sec:calculate}), we obtain the Choi state of the conditional expectation $\mathbb{E}_B$ on $(\cM_B^S)'$ by using this characterization. A conditional expectation is an idempotent positive linear map onto a subalgebra, and thus it is a projection as a superoperator (an operator acting on the vector space of operators). The projection onto the fixed point algebra $\cF(\cT_{B})$ is obtained as the projection with eigenvalue 1 in the spectral decomposition of the superoperator representation of $\cT_B$. The invariance under the modular transformation for all $t$ is equivalent to the invariance under the adjoint action ${\rm ad}_{\rho_B}$ of the generator of $\Delta_{\rho_B}^t$, from which we can remove the dependence on $t$. Hence, the projection onto $(\cM_B^S)'$ is obtained from the superoperators corresponding to $\cT_B$~\eqref{def:E_T} and ${\rm ad}_{\rho_B}$~\eqref{def:LR}.

 The Choi state of $\mathbb{E}_B$ is obtained from the superoperator projection via the reshuffling~\cite{Bengtsson_Zyczkowski_2006}. Its canonical purification consists of the direct sum of four maximally entangled states among four copies of $B_i^R, B_i^L$ (Fig.~\ref{fig:Cstates} in Sec.~\ref{sec:calculate}). Simply calculating the rank of $B$ provides $\sum_id_{B_i^R}d_{B_i^L}=d_B$, counting contributions from both the entanglement between $B_i^R$ and its copy and the entanglement between $B_i^L$ and its copy. 
 We then extract the optimal dimension~\eqref{eq:mindim} by only removing the entanglement between $B_i^R$ and its copy via an optimization over bipartite unitaries. This final step enables us to count the contribution only from $B_i^L$. 

\subsection{The statement of the main result}
To state the main theorem, we start from the following lemma (Lemma~\ref{lem:Minimals} (iv) in Sec.~\ref{sec:3}).
\begin{lem*}
 Define $\cT_B:=\cR^{\tau,\Omega_{E\to B}}_{E\to B}\circ\Omega_{B\to E}$, the composition of $\Omega_{B\to A}^c$ (the complementary map of $\Omega_{B\to A}:=(\Omega^\dagger_{A\to B})^\dagger$) and the Petz recovery map $\cR^{\tau,\Omega_{E\to B}}_{E\to B}$. Then, $(\cM_B^S)'$ is the largest subalgebra of $\cF(\cT_B)$ which is invariant under ${\rm ad}_{\rho_B}(\cdot)=[\cdot,\log\rho_B]$.
\end{lem*}

We are going to calculate the conditional expectation onto $(\cM_B^S)'$, which we denote by $\mathbb{E}_B$, from $\rho_{AB}$. To do so, we first show a way to calculate the supeoperator representation of $\cT_B$. Let us denote the nonzero eigenvalues of $J_{AB}$~\eqref{def:Choi1} as $\omega_i$. Subsequently, the Kraus operators $\{K_i^\dagger\}$ of $\Omega_{A\to B}^\dagger(\cdot)=\sum_iK_i^\dagger\cdot K_i$, which is defined in Eq.~\eqref{def:Omega1}, satisfy $\tr(K_i^\dagger K_j)=\omega_i\delta_{ij}$: 
We then define a deformed CP-map ${\tilde \Omega}^\dagger_{A\to B}$ as 
\begin{equation}\label{def:tOmega}
    {\tilde \Omega}^\dagger_{A\to B}(\cdot)=\sum_{i=1}^{\rank(\rho_{AB})}\omega^{-\half}_iK_i^\dagger\cdot K_i\,.
\end{equation}
${\tilde \Omega}^\dagger_{A\to B}$ is no longer unital but CP. The Choi operator of ${\tilde \Omega}^\dagger_{A\to B}$ is given by $\sqrt{J_{AB}}$. Then the superoperator representation of $\cT$, $E_\cT$, is given as
\begin{equation}\label{def:E_T}
      E_\cT=\sum_{a,b=1}^{d_A}{\tilde \Omega}^\dagger_{A\to B}(|a\rag\lag b|)\otimes\overline{ {\tilde \Omega}^\dagger_{A\to B_1}(|a\rag\lag b|)}\,, 
\end{equation}
where $\cH_{B_1}\cong \cH_{B}$ and $\{|a\rag\}$ is an orthonormal basis of $\cH_A$  (see Sec.~\ref{sec:ETcal} for the proof).

The linear map ${\rm ad}_{\rho_B}(\cdot)$ also has the superoperator representation as 
\begin{align}
RL_B:=I_B\otimes\log\rho_{B_1}^{T_{B_1}}-\log\rho_B\otimes I_{B_1}\label{def:LR}\,.
\end{align}
$RL_B$ corresponds to the generator of the modular transformation $\Delta_{\rho_B}^{t}$, which plays an important role in the theory of sufficiency~\cite{Petz1986-nm,Petz1988-ht}.

The conditional expectation $\mathbb{E}_B$ is then obtained via the spectral decompositions of these operators:
\begin{align}
   E_\cT&=\bigoplus_{\lambda}\lambda P_\lambda\label{eq:specET}\\
   RL_B&=\bigoplus_\eta\eta Q_\eta\,\label{eq:specRL}\,,
\end{align}
where $P_\lambda$ and $Q_\eta$ are orthogonal projections to the eigensubspaces of $E_\cT$ and $RL_B$, respectively.
Especially, $P_1 \:(\lambda=1)$ is the projection onto the space of the fixed-points of $\cT_B$.
Denote the set of eigenvalues of $A$ by ${\rm spec}(A)$ and define $P_V$ as the projector onto the subspace
\begin{equation*}
    V:= \bigoplus_{\eta\in{\rm spec}(RL)}\left(\supp(Q_\eta)\cap \supp(P_1)\right)\,,
\end{equation*}
which is the subspace of the fixed point of $\cT$ invariant under $\Delta_{\rho_B}^{t}$. 
Using the formula given in ~\cite{ANDERSON1969576}, the projection onto $V$ is
\begin{equation*}
    P_V= 2\bigoplus_{\eta\in{\rm spec}(RL)}Q_\eta (Q_\eta+P_1)^{-1}P_1\,,
\end{equation*}
where $^{-1}$ is the Moore-Penrose inverse. From the lemma in the above, $P_V$ is the superoperator representation of $\mathbb{E}_B$. 

Consider systems $\cH_{\bar B}\cong \cH_{\bar B_1}\cong \cH_B$ and define $|\cI\rrag_{BB_1}:=\sum_i|ii\rag_{BB_1}$ in the basis of the complex conjugate/transposition in Eq. ~\eqref{def:E_T} and \eqref{def:LR}. 
The normalized Choi state $C_{BB_1}$ of $\mathbb{E}_{B}$ is defined as follows:
\begin{align}
C_{BB_1}&:=\frac{1}{d_B}(\id_{B}\otimes\mathbb{E}_{B_1})\left(|\cI\rrag\llag\cI|_{BB_1}\right) \nonumber \\
&=\bigoplus_i \frac{d_{B_i^L}d_{B_i^R}}{d_B} \tau_{B_i^L}\otimes |\Psi_i\rrag\llag\Psi_i|_{B_i^R{B_{1i}^R}}\otimes \tau_{B_{1i}^L}\label{eq:Cbb1} \,,
\end{align}
where each $|\Psi_i\rrag$ is a maximally entangled state. 
$C_{BB_1}$ and $P_V$ are related via the reshuffling map~\cite{Bengtsson_Zyczkowski_2006}:
\begin{equation*}
    C_{BB_1}=\frac{1}{d_B}\sum_{i,j=1}^{d_B}(I_B\otimes|i\rag\lag j|_{B_1})P_V(|i\rag\lag j|_{B}\otimes I_{B_1}).
\end{equation*}

The next step is extracting the optimal dimension~\eqref{eq:mindim} from $C_{BB_1}$. Consider the canonical purification of $C_{BB_1}$, which is the purification in an eigenbasis of $C_{BB_1}$, denoted as $|C\rag_{BB_1{\bar B}{\bar B_1}}$.
We then optimize over all unitaries on ${\bar B}{\bar B}_1$ to minimize the entanglement entropy between $B{\bar B}$ and $B_1{\bar B_1}$:
\begin{equation}\label{eq:minU}
    {\tilde U}_{{\bar B\bar B_1}}:=\argmin_U \:S(B{\bar B})_{U_{{\bar B}{\bar B_1}}|C\rag}\,.
\end{equation}
The optimization can be performed using, for example, a gradient algorithm~\cite{PhysRevLett.126.120501}.

The main theorem of this work is showing that the minimal dimension of exact local compression is given by the Schmidt rank of ${\tilde U}_{{\bar B\bar B_1}}|C\rag_{BB_1{\bar B}{\bar B_1}}$. The proof is  presented in Sec.~\ref{sec:calculate}.
\begin{thm}\label{thm:main}
For any $\rho_{AB}$, the minimal dimension of any exact local compression is given by
    \begin{equation}\label{eq:minimal dimension}
        d_\tB^{\min}=\sum_{i}d_{B^L_i}=\Sch\left(BB_1{\bar B_1}:{\bar B}\right)_{|\tilde C\rag}\,,
    \end{equation}
     where $\cH_B\cong\bigoplus_i\cH_{B_i^L}\otimes\cH_{B_i^R}$ is the Koashi-Imoto decomposition of $\cH_B$ and
    \begin{equation}
        |{\tilde C}\rag_{BB_1{\bar B}{\bar B_1}}:={\tilde U}_{{\bar B\bar B_1}}|C\rag_{BB_1{\bar B}{\bar B_1}}\,.
    \end{equation}
It also holds that
    \begin{equation}
        d_{B^R}:=\sum_id_{B^R_i}=\Sch\left(BB_1:{\bar B}{\bar B_1}\right)_{|\tilde C\rag}\,.
    \end{equation}
\end{thm}
From the expression of $C_{BB_1}$ in \eqref{eq:Cbb1}, the following holds:
\begin{cor}
\begin{equation}\label{cor:bound}
    \rank(C_{BB_1})=\sum_{i}d_{B^L_i}^2\,
\end{equation}
and
    \begin{equation}
        \sqrt{\rank(C_{BB_1})}\leq d_\tB^{\min}\leq \rank(C_{BB_1})\,.
    \end{equation}
\end{cor}
Unlike Eq.~\eqref{eq:minimal dimension}, calculating Eq.~\eqref{cor:bound} does not require any optimization.

\section{Characterizing minimal sufficient subalgebra}\label{sec:3}
We examine the details of the algebras $\cM_B^S$ and $(\cM_B^S)'$.  
Originally, Petz demonstrated that $\cM_B^S$ has the form~\cite{Petz1986-nm}
\begin{equation}
    \cM_B^S=\alg\left\{\mu^{it}_B\rho^{-it}_B, \mu\in\cS, t\in\mathbb{R} \right\}\,
\end{equation}
in a finite-dimensional system. More recently, 
Ref.~\cite{Jencova2012-sy} showed that it can also be written as
\begin{equation}
    \cM_B^S=\alg\left\{\rho^{it}_Bd(\mu,\rho)\rho^{-it}_B, \mu\in\cS, t\in\mathbb{R} \right\}\,,
\end{equation}
where $d(\mu,\rho):=\rho^{-\half}\mu\rho^{-\half}$ is the {\it commutant Radon-Nikodym derivative} introduced in Ref.~\cite{Jencova2012-sy}. These alternative generators are positive semi-definite and easy to treat sometimes. 

By inserting the definition of $\mu_B$~\eqref{def:mu}, $d(\mu,\rho)$ can be rewritten as
\begin{align}
d(\mu,\rho)&=\rho_{B}^{-\half}\frac{\tr_A\left(M_A\rho_{AB}\right)}{\tr\left(M_A\rho_{A}\right)}\rho_{B}^{-\half}\nonumber\\
    &=\frac{1}{\tr\left(M_A\rho_{A}\right)}\tr_{A}\left(\rho_{B}^{-\half}\rho_{AB}\rho_{B}^{-\half}(M_A\otimes I_B)\right)\nonumber\\
    &=\frac{1}{\tr\left(M_A\rho_{A}\right)}\tr_{A}\left(J_{AB}(M_A\otimes I_B)\right)\nonumber\\
    &=\frac{\Omega^\dagger_{A\to B}(M_A^{T_A})}{\tr\left(M_A\rho_{A}\right)}\,,\label{eq:cRNd1}
\end{align}
where $J_{AB}$ is defined by Eq.~\eqref{def:Choi1} and $\Omega_{A\to B}^\dagger$ is defined in Eq.~\eqref{def:Omega1}. 

It might be easier to characterize $(\cM_B^S)'$ than $\cM_B^S$ in our setting. Denote the complementary channel of $\Omega_{B\to A}$ by $\Omega_{B\to E}$.  Recall that $\cF(\cE^\dagger)$ is the fixed-point algebra of a unital CP-map $\cE^\dagger$ and $\cA_\cE$ is the correctable algebra of a CPTP-map $\cE$.  We then obtain the following characterizations: 
\begin{lem}\label{lem:Minimals}
$(\cM_B^S)'$ is equivalent to 
\begin{itemize}
    \item[(i)] The intersection of the correctable algebras 
    \begin{equation}
        \bigcap_{t\in\mathbb{R}}\cA_{\left(\Omega_{B\to E}\circ\Delta_{\rho_B}^t\right)}
    \end{equation} 
    \item[(ii)] The intersection of the fixed-point algebras 
    \begin{equation}
        \bigcap_{t\in\mathbb{R}}\cF\left(\Delta_{\rho_B}^{-t}\circ\cR_{E\to B}^{\tau,\Omega}\circ\Omega_{B\to E}\circ\Delta_{\rho_B}^t\right) 
    \end{equation} 
    \item[(iii)]
The largest subalgebra of $\cF(\cR_{E\to B}^{\tau,\Omega}\circ\Omega_{B\to E})$ that is invariant under $\Delta^t_{\rho_B}$ for any $t\in\mathbb{R}$.
    \item[(iv)]
The largest subalgebra of $\cF(\cR_{E\to B}^{\tau,\Omega}\circ\Omega_{B\to E})$ that is invariant under ${\rm ad}_{\rho_B}(\cdot)$: 
\begin{equation}
    {\rm ad}_{\rho_B}((\cM_B^S)')=[(\cM_B^S)',\log\rho_B]\subset(\cM_B^S)'\,.
\end{equation}
\end{itemize}
Note that $\cR_{E\to B}^{\tau,\Omega}$ is the Petz recovery map with respect to $\tau_B$ for $\Omega_{B\to E}$.
\begin{equation}\label{eq:PetzmapQEC}
    \cR^{\tau,\Omega}_{E\to B}(\cdot)=\Omega_{B\to E}^\dagger\circ\Gamma_{\Omega_{B\to E}(I_{B})}^{-1}(\cdot)\,.
\end{equation}

\end{lem}
\begin{proof}
From Eq.~\eqref{eq:cRNd1}, we obtain 
   \begin{align*}
  &X_B\in(\cM_B^S)'  \\
  &\Leftrightarrow [X_B, \rho^{it}_Bd(\mu,\rho)\rho^{-it}_B]=0\,,\;\; \forall t\in\mathbb{R}, \forall \mu_B\in\cS \,,\\
    &\Leftrightarrow \left[X_B, \Delta_{\rho_B}^t\circ \Omega_{A\to B}^\dagger(M_A^{T_A}) \right]=0\,,\;\;\forall t\in\mathbb{R},\: 0\leq \forall M_A\leq I_A\,,\\
    &\Leftrightarrow X_B\in\left({\rm Im} \left(\Delta_{\rho_B}^t\circ \Omega_{A\to B}^\dagger\right) \right)' \,,\;\;\forall t\in\mathbb{R}\,.
\end{align*}
Then, we use a result from operator-algebra quantum error correction theory:
\begin{lem}\label{lem:ac=com}~\cite{Beny2007-ua}
For any CPTP-map $\cE$, let $\cA_{\cE^c}$ be the correctable algebra of the complementary map $\cE^c$. Then, it holds that
    \begin{equation}
        \cA_{\cE^c}=({\rm Im}\cE^\dagger)'\,.
    \end{equation}
\end{lem}
$(\Delta_{\rho_B}^{t}\circ\Omega_{A\to B}^\dagger)^\dagger=\Omega_{B\to A}\circ\Delta_{\rho_B}^{-t}$, and a simple calculation shows that
\begin{equation*}
    \left(\Omega_{B\to A}\circ\Delta_{\rho_B}^t\right)^c=\Omega_{B\to A}^c\circ\Delta_{\rho_B}^t=\Omega_{B\to E}\circ\Delta_{\rho_B}^t\,.
\end{equation*}
Hence, Lemma~\ref{lem:ac=com} implies that 
\begin{equation*}
    \left({\rm Im} \left(\Delta_{\rho_B}^t\circ \Omega_{A\to B}^\dagger\right) \right)'=\cA_{\left(\Omega_{B\to A}\circ\Delta_{\rho_B}^{-t}\right)^c},
\end{equation*}
which completes step $(i)$. To show the equivalence of $(ii)$, we utilize the following facts regarding correctable algebra:
\begin{lem}\cite{Chen2020-in}\label{lem:OAQECFix} For any CPTP-map $\cE$, it holds that 
\begin{equation}    \cA_\cE=\cF\left(\cR^{\tau,\cE}\circ\cE\right)=\cF\left(\cE^\dagger\circ(\cR^{\tau,\cE})^\dagger\right) \,.
\end{equation}
\end{lem}
This lemma implies
\begin{equation}\label{eq:lemma2impli}
    \cA_{\left(\Omega_{B\to A}\circ\Delta_{\rho_B}^t\right)^c}=\cF\left(\cR_{E\to B}^{\tau,\Omega\circ\Delta^t}\circ\Omega_{B\to E}\circ\Delta_{\rho_B}^t \right)\,.
\end{equation}

The Petz recovery map $\cR_{E\to B}^{\tau,\Omega\circ\Delta^t}$ with respect to $\tau_B$ for $\Omega_{B\to E}\circ\Delta_{\rho_B}^t$ is simplified as follows:
\begin{align}
    &\cR_{E\to B}^{\tau,\Omega\circ\Delta^t}(\cdot)\nonumber\\
    &:=\left(\Omega_{B\to E}\circ\Delta^t_{\rho_B}\right)^\dagger\left(\Omega_{B\to E}(I_{B})^{-\half}\cdot\Omega_{B\to E}(I_{B})^{-\half}\right)\nonumber\\
    &=\Delta^{-t}_{\rho_B}\left(\Omega_{E\to B}^\dagger\left(\Omega_{B\to E}(I_{B})^{-\half}\cdot\Omega_{B\to E}(I_{B})^{-\half}\right)\right)\,\nonumber\\
    &=\Delta^{-t}_{\rho_B}\circ \cR_{E\to B}^{\tau,\Omega}(\cdot)\,.
\end{align}
Combining this with Eq.~\eqref{eq:lemma2impli} completes $(ii)$. 

To see $(ii)\Leftrightarrow(iii)$, $X_B \in (ii)$ is equivalent to 
\begin{align}
    \Delta_{\rho_B}^{-t}\circ\cR_{E\to B}^{\tau,\Omega}\circ\Omega_{B\to E}\circ\Delta_{\rho_B}^t(X_B)=X_B\,,\: \forall t\in\mathbb{R}&\nonumber\\
    \Leftrightarrow\:\cR_{E\to B}^{\tau,\Omega}\circ\Omega_{B\to E}(X_B(t)) =X_B(t)\,,\:\forall t\in\mathbb{R}&\,,\label{eq:time_invariant}
\end{align}
where $X_B(t):=\Delta_{\rho_B}^t(X_B)$. Therefore, the algebra $(ii)$ is equivalent to a modular invariant subalgebra in $\cF(\cR_{E\to B}^{\tau,\Omega}\circ\Omega_{B\to E})$. It is largest since any element of other modular invariant subalgebra must satisfies Eq.~\eqref{eq:time_invariant}. 

Lastly, we show $(iii)\Leftrightarrow(iv)$. For any $X_B\in(\cM_B^S)'$, Eq.~\eqref{eq:time_invariant} implies 
\begin{align*}
    \cR_{E\to B}^{\tau,\Omega}\circ\Omega_{B\to E}&\left(X_B(t+\delta t)-X_B(t)\right)\\
   &=X_B(t+\delta t)-X_B(t)\,,\:\forall t,\delta t\in\mathbb{R}\,.
\end{align*}
Therefore derivatives
\begin{align*}
    \frac{d^nX_B(t)}{dt^n}&=-i\left[\frac{d^{n-1}X_B(t)}{dt^{n-1}}, \log\rho_B \right]\\
    &=(-i)^n[...[[X_B(t),\log\rho_B],\log\rho_B],...,\log\rho_B]\\
    &=(-i)^n {\rm ad}_{\rho_B}^n(X_B(t))
\end{align*}
are the fixed points of $\cR_{E\to B}^{\tau,\Omega}\circ\Omega_{B\to E}$ for all $n\in\mathbb{N}$. Since
$\Delta^t_{\rho_B}(ad^n_{\rho_B}(X_B))=ad^n_{\rho_B}(X_B(t))$, $ad^n_{\rho_B}(X_B)$ satisfies Eq.~\eqref{eq:time_invariant}, for all $n\in\mathbb{N}$. Therefore, ${\rm ad}_{\rho_B}(X_B)\in(\cM_B^S)'$ and $[(\cM_B^S)',\log\rho_B]\subset(\cM_B^S)'$. 

Suppose that there is a subalgebra of $\cN_B\subset\cF(\cR_{E\to B}^{\tau,\Omega}\circ\Omega_{B\to E})$ satisfying $[\cN_B,\log\rho_B]\subset\cN_B$. Then, for any $X_B\in\cN_B$, the iterated commutators ${\rm ad}_{\rho_B}^n(X_B)$ are also in $\cF(\cR_{E\to B}^{\tau,\Omega}\circ\Omega_{B\to E})$ for all $n\in\mathbb{N}$. 
This implies that
\begin{equation*}
    X_B(t)=\sum_{n=0}^\infty \frac{d^nX_B(t)}{dt^n}\vert_{t=0}t^n=\sum_{n=0}^\infty (-i)^n{\rm ad}_{\rho_B}^n(X_B)t^n
\end{equation*}
is also in $\cF(\cR_{E\to B}^{\tau,\Omega}\circ\Omega_{B\to E})$ for all $t\in\mathbb{R}$; thus, $X_B\in(\cM_B^S)'$ by Eq.~\eqref{eq:time_invariant}.
\end{proof}

If $(\cM^S_B)'$ is abelian, it is isomorphic to $\mathbb{C}^k$ for some $k\in\mathbb{N}$. This is a direct sum of one-dimensional systems $\mathbb{C}$, and hence one can write 
\begin{equation*}
    (\cM^S_B)'\cong\bigoplus_i I_{\cH_{B_i^L}}\ot \mathbb{C}\,,
\end{equation*}
which implies $d_{B_i^R}=1$ in Eq.~\eqref{eq:decHBminimal} for all $i$. Thus,
\begin{equation*}
d_B=\sum_id_{B_i^L}d_{B_i^R}=\sum_id_{B_i^L}=d_{B^L}\,,    
\end{equation*}
i.e., nontrivial exact compression is impossible. To achieve nontrivial compression, $(\cM_B^S)'$ must be non-abelian.  Lemma~\ref{lem:Minimals}(i) for $t=0$ implies that $(\cM_B^S)'$ is a subalgebra of $\cA_{\Omega_{B\to E}}$ Therefore, we have the following criteria, which we set as an independent theorem due to its usefulness.
\begin{thm}\label{thmET}
It holds that 
\begin{equation}\label{eq:cAthmFT}
\cA_{\Omega_{B\to E}}=\cF\left(\cR_{E\to B}^{\tau,\Omega}\circ\Omega_{B\to E}\right)=({\rm Im}\Omega_{A\to B}^\dagger)'
\end{equation}
and a necessary condition for $d_\tB^{\min}<d_B$ is that the algebra~\eqref{eq:cAthmFT} is non-abelian.
\end{thm}

\section{Calculation of the minimal dimension}\label{sec:calculate}
We provide a method for calculating the minimal dimensions, as shown in Theorem~\ref{thm:main}. 
Recall that we define $\cT_{B\to B}:=\cR_{E\to B}^{\tau,\Omega}\circ\Omega_{B\to E}$ which is a self-dual CPTP-map.

\subsection{Proof of Eq.~\eqref{def:E_T}}\label{sec:ETcal}
Let $\{K_i^\dagger\}$ be the Kraus operators of $\Omega_{A\to B}^\dagger(\cdot)=\sum_i K_i^\dagger\cdot K_i$ satisfying the orthogonality condition $\tr(K_i^\dagger K_j)=\omega_i\delta_{ij}$ (which always exists), where $\omega_i$ is a non-zero eigenvalue of $J_{AB}$. 
The Kraus operators of $\cT_{B\to B}$ are calculated in more detail. By definition, it holds that:
\begin{equation}\label{eq:Tform}
    \cT_{B\to B}=\cR_{E\to B}^{\tau,\Omega}\circ\Omega_{B\to E}=\Omega^\dagger_{E\to B}\circ\Gamma^{-1}_{O_E}\circ\Omega_{B\to E}\,,
\end{equation}
where $O_E:=\Omega_{B\to E}(I_B)$ is a positive operator (we set $d_E=\rank(\rho_{AB})$). For  fixed bases of $A$ and $E$ denoted by $\{|a\rag_A\}$ and $\{|\phi_i\rag_E\}$, respectively, we define
\begin{equation}\label{eq:OmBEKraus}
    F_a:=\sum_{i=1}^{\rank(\rho_{AB})}|\phi_i\rag\lag a|K_i\,, 
\end{equation}
which provides a Kraus representation of $\Omega_{B\to E}=\Omega_{B\to A}^c$. Owing to the orthogonality condition, we obtain
\begin{align}
    O_E&:=\Omega_{B\to E}(I_B)\nonumber\\
    &=\sum_aF_aF_a^\dagger\nonumber\\
    &=\sum_{i,j}\tr(K_j^\dagger K_i)|\phi_i\rag\lag\phi_j|\nonumber\\
    &=\sum_{i}\omega_i|\phi_i\rag\lag\phi_i|,\label{eq:omegaspec}
\end{align}
such that $O_E$ is diagonal. Combining Eqs.~\eqref{eq:Tform}~\eqref{eq:omegaspec}, the Kraus operators of $\cT$ can be written as
\begin{align}
    T_{(a,b)}&=F_a^\dagger O_E^{-\half}F_b\nonumber\\
    &=\sum_{i,j}K_i^\dagger|a\rag\lag \phi_i|O_E^{-\half}|\phi_j\rag\lag b|K_j\nonumber\\
    &={\tilde \Omega}^\dagger_{A\to B}(|a\rag\lag b|)\,,\label{eq:KrausT}
\end{align}
where ${\tilde \Omega}$ is defined by Eq.~\eqref{def:tOmega}.  

The superoperator of $\cT$
\begin{equation}\label{eq:superET2}
   E_\cT:=\sum_{(a,b)}T_{(a,b)}\otimes {\bar T}_{(a,b)}\in\cB\left(\cH_B\otimes\cH_{B_1}\right)\,,
\end{equation}
can now be written as:
\begin{equation}\label{eq:superET1}
    E_\cT=\sum_{a,b}{\tilde \Omega}_{A\to B}(|a\rag\lag b|)\otimes\overline{ {\tilde \Omega}_{A\to B_1}(|a\rag\lag b|)}\,.
\end{equation}

 Note that Eq.~\eqref{eq:KrausT} and Theorem~\ref{thmET} imply that:
 \begin{equation}
     ({\rm Im}{\tilde \Omega}_{A\to B}^\dagger)'=\cF(\cT)=({\rm Im}{\Omega}_{A\to B}^\dagger)' \,.
 \end{equation}

\subsection{Vectorization of $(\cM_B^S)'$}
Introducing vectorization of the operators. 
\begin{equation}
    X_B\mapsto |X_B\rrag:=(X_B\ot I_{B_1})|\cI\rrag_{B{B_1}}\,,\quad \cH_{B_1}\cong \cH_B\,.
\end{equation}
The action of $\cT$ on $X_B$ corresponding to that of $E_\cT$ on $|X_B\rrag$, 
and the action of ${\rm ad}_{\rho_B}$ is also rephrased by the superoperator 
\begin{equation}
    RL_{B}=I_B\otimes\log\rho_{B_1}^{T_{B}}-\log\rho_B\otimes I_{B_1}\,,
\end{equation}
where the transposition is taken in the same basis as the conjugation in Eq.~\eqref{eq:superET1}. 

Then, we have a vector representation of the elements in $(\cM_B^S)'$ as follows:
\begin{lem}
Consider the following spectral decomposition: 
\begin{align}
   E_\cT&=\bigoplus_{\lambda}\lambda P_\lambda\,,\\
   RL_{B}&=\bigoplus_\eta\eta Q_\eta\,.
\end{align}
Then it holds that
\begin{equation}
    V:= \bigoplus_{\eta\in{\rm spec}(RL_B)}(Q_\eta\cap P_1)\,
\end{equation}
is the vectorization of $(\cM_B^S)'$. 
\end{lem}
\begin{proof}
By definition, $V$ is spanned by the eigenvectors of $RL_{B}$, stabilized by $P_1$. 
Moreover, this is the largest subspace of $\supp(P_1)$ such that
\begin{align}
    E_\cT|X_B\rrag&=|X_B\rrag\,,\label{eq:prooflem3-1}\\
    RL_B|X_B\rrag&\in V\label{eq:prooflem3-2}
\end{align}
for all $|X_B\rrag\in V$.  This is because Eq.~\eqref{eq:prooflem3-2} implies $RL_B=P_VRL_BP_V\oplus (I-P_V)RL_B(I-P_V)$, which further implies that $V$ is spanned by the eigenvectors of $RL_B$ (recall that $P_V$ is the projection onto $V$). Eq.~\eqref{eq:prooflem3-1} implies that these eigenvectors must be in $\supp(P_1)$. 

Lemma~\ref{lem:Minimals}(iv) clarifies that the vectorization of $(\cM_B^S)'$ must be a subspace of $V$. By inverting the vectorization of $V$, we obtain an operator subspace $\cV$ that is fixed by $\cT_{B\to B}$ and is invariant under ${\rm ad}_{\rho_B}$. Therefore, for any $X_B\in\cV$, $\cT_{B\to B}(X_B)=X_B$ and ${\rm ad}_{\rho_B}^n(X_B)\in \cV$ $\forall n\in\mathbb{N}$. The latter condition implies $X_B(t)\in \cV\subset\cF(\cT_{B\to B})$ $\forall t\in\mathbb{R}$. Therefore, Lemma~\ref{lem:Minimals}(iii) yields $\cV\subset(\cM_B^S)'$.
\end{proof}
Because we know that $V$ is a vectorization of $(\cM_B^S)'$, we also obtain the superoperator of the conditional expectation of $(\cM_B^S)'$. 
\begin{cor}
The projection operator, $P_V$ on $V$, is the superoperator representation of the conditional expectation (with respect to the completely mixed state) on $(\cM_B^S)'$.   
\end{cor}

\subsection{Proof of Theorem~\ref{thm:main}}
We denote the conditional expectation of $(\cM_B^S)'$ as $\mathbb{E}_B$. 
Recall that there exists a decomposition $\cH_B\cong\bigoplus_i\cHil\otimes\cHir$ such that
\begin{align*}
(\cM_B^S)'\cong &\bigoplus_i I_{B_i^L}\otimes \cB(\cHir)\,.
\end{align*}
The conditional expectation $\mathbb{E}_{B}$ (with respect to $\tau_B$) is then written as
\begin{align*}
\mathbb{E}_{B}(\cdot)=&\bigoplus_i \tau_{B_i^L}\otimes \tr_{\cHil}\left(\Pi_i\cdot\Pi_i\right)\,,
\end{align*}
where $\Pi_i$ is a projection onto $\cHil\otimes\cHir$. 
Up to a local unitary, $|\cI\rrag_{BB_1}$ also exhibits decomposition
\begin{equation}\label{eq:IBB}
    |\cI\rrag_{BB_1}=\sum_i|\cI_i\rrag_{B_i^LB_{1i}^L}\otimes |\cI_i\rrag_{B_i^R{B_{1i}^R}}\,\,.
\end{equation}
Note that the argument in the following is independent of the basis choice in $|\cI\rrag_{BB_1}$. 
The normalized Choi state has the following corresponding decomposition:
\begin{align}
  C_{BB_1}&:=\frac{1}{d_B}(\id_{B}\otimes\mathbb{E}_{B_1})\left(|\cI\rrag\llag\cI|_{BB_1}\right)\\
  &=\frac{1}{d_B}\bigoplus_i  \tau_{B_i^L}\otimes |\cI_i\rrag\llag\cI_i|_{B_i^R{B_{1i}^R}}\otimes I_{B_{1i}^L}\,\\
  &=\bigoplus_i p_i \tau_{B_i^L}\otimes |\Psi_i\rrag\llag\Psi_i|_{B_i^R{B_{1i}^R}}\otimes \tau_{B_{1i}^L}\,,
\end{align}
where $p_i:=d_{B_i^L}d_{B^R_i}/d_B$, the square root of which is: 
\begin{align}
  \sqrt{C_{BB_1}}=\bigoplus_i \sqrt{p_i} \sqrt{\tau}_{B_i^L}\otimes |\Psi_i\rrag\llag\Psi_i|_{B_i^R{B_{1i}^R}}\otimes \sqrt{\tau}_{B_{1i}^L}\,.
\end{align}

Next, we consider the canonical purification of $C_{BB_1}$.
\begin{align*}
  |C&\rag_{BB_1{\bar B}{\bar B_1}}\\
  &:=\sqrt{C_{BB_1}}|\cI\rrag_{B{\bar B}}|\cI\rrag_{B_1{\bar B_1}}\\
  &=\sum_i  \sqrt{p_i} |\Psi_i\rrag_{B_i^L{\bar B}_{i}^L}|\Psi_i\rrag_{B_i^R{B_{1i}^R}}|\Psi_i\rrag_{{\bar B}_{i}^R{{\bar B}_{1i}^R}}|\Psi_i\rrag_{B_{1i}^L{\bar B}_{1i}^L}\,,
\end{align*}
where $\cH_{\bar B}\cong\cH_{{\bar B}_1}\cong \cH_B$. 
Applying a unitary $U_{{\bar B}{\bar B}_1}$ gives: 
\begin{align*}
  U_{{\bar B}{\bar B}_1}|C&\rag_{BB_1{\bar B}{\bar B_1}}=\sum_i \sqrt{p_i} |\Psi_i\rrag_{B_i^R{B_{1i}^R}}|\Phi_i\rrag_{B_i^L{\bar B}{\bar B}_1B_{1i}^L}\,
\end{align*}
where $|\Phi_i\rrag_{B_i^L{\bar B}{\bar B}_1B_{1i}^L}$ denotes some pure state. 
The reduced state of $B{\bar B}$ is then calculated as
\begin{equation}
    \bigoplus_i p_i \tau_{B_i^R}\otimes \Phi_{i, B^L_i{\bar B}}\,
\end{equation}
and the entanglement entropy is calculated as
\begin{equation}
    S(B{\bar B})_{UCU^\dagger}=H(\{p_i\})+\sum_ip_i\left(\log d_{B^R_i}+S(B^L_i{\bar B})_{\Phi_i}\right)\,,
\end{equation}
where $H(\cdot)$ denotes the Shannon entropy. Only the final term depends on $U_{{\bar B}{\bar B}_1}$. 
The minimum in Eq.~\eqref{eq:minU} is achieved when every $\Phi_{i, B^L_i{\bar B}}$ is in a pure state (the minimum value is the entanglement of the purification~\cite{ent_puri} of $C_{BB_1}$). 
The corresponding state is given as 
\begin{align}
  &|{\tilde C}\rag_{BB_1{\bar B}{\bar B_1}}=\sum_i \sqrt{p_i}  |\Psi_i\rrag_{B_i^R{B_{1i}^R}}
  |\Psi_i\rrag_{B_i^L{\bar B}}|\Psi_i\rrag_{B_{1i}^L{\bar B}_1} \,,
\end{align}    
where $|\Psi_i\rrag_{B_i^L{\bar B}}$ (and $|\Psi_i\rrag_{B_{1i}^L{\bar B}_1}$) are pure maximally entangled states between $B_i^L$ and ${\bar B}$ ($B_i^R$ and ${\bar B}_1$)(Fig. ~\ref{fig:Cstates}). 

It is clear that 
\begin{align}
    \Sch(BB_1{\bar B}_1:{\bar B})_{\tilde C}&=\rank({\tilde C}_{\bar B})=\sum_i{d_{B^L_i}}\\
    \Sch(B{\bar B}:B_1{\bar B}_1)_{\tilde C}&=\rank({\tilde C}_{{\bar B}{\bar B}_1})=\sum_i{d_{B^R_i}}
\end{align}
By comparing $|{\tilde C}\rag_{BB_1{\bar B}{\bar B_1}}$ and decomposition~\eqref{eq:minimaldecoH}, we conclude $\cH_{\tB}=\bigoplus_i\cH_{B_i^L}$ from which the theorem holds. 

Note that it might be true that the minimal achievable dimension can be calculated by minimization:
\begin{equation}\label{eq:onlyCbb1}
\min_{U_{BB_1}}\rank\left(\tr_{B_1}\left(U_{BB_1}C_{BB_1}U_{BB_1}^\dagger\right)\right)\,.
\end{equation}
This expression does not need a purification, however, we do not know an easily implementable algorithm to perform the minimization in Eq.~\eqref{eq:onlyCbb1}.  

\begin{figure}[htbp]
	\centering
	\includegraphics[width=0.45\textwidth]{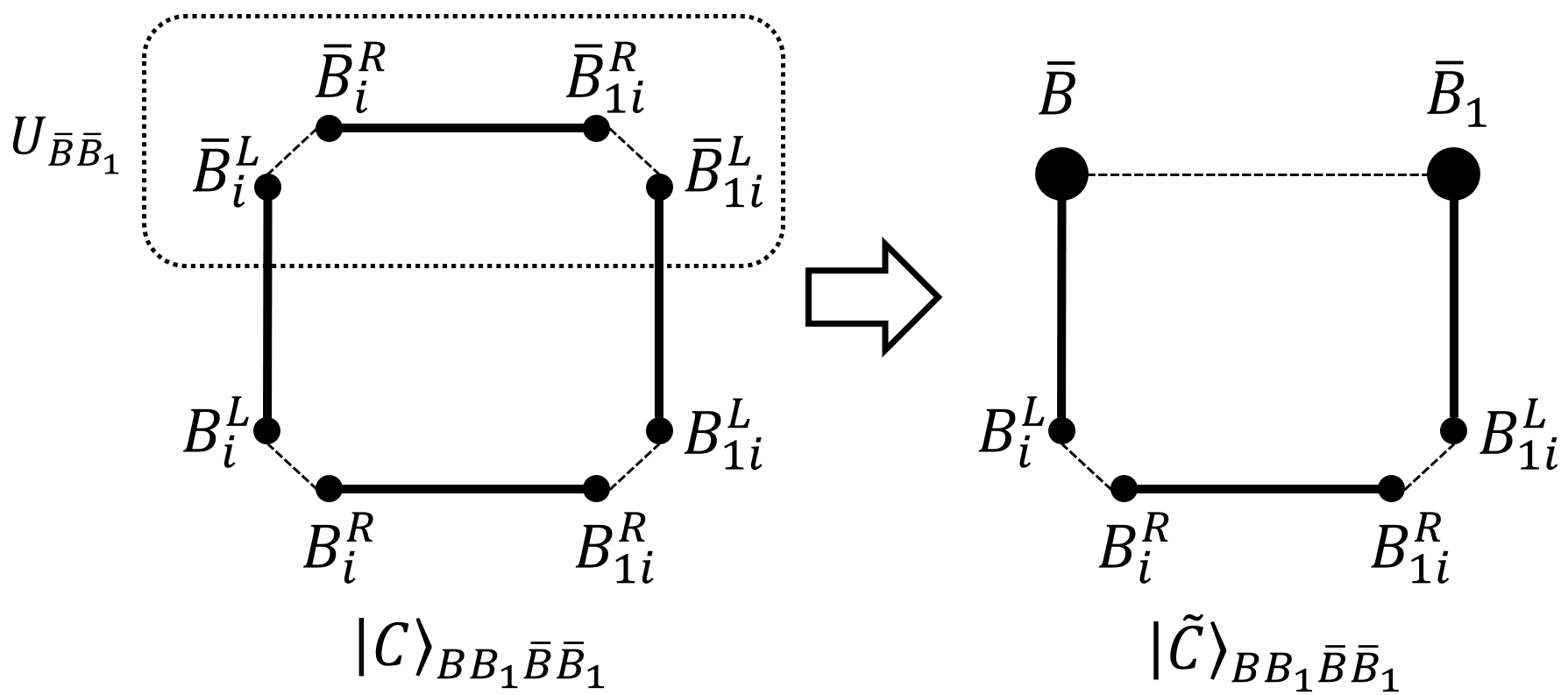}
	\caption{(Left) Canonical purification of $C_{BB_1}$, which is the sum of tensor products of four maximally entangled states. By minimizing the entanglement between $B{\bar B}$ and $B_1{\bar B}_1$ by the unitary $U_{{\bar B}{\bar B}_1}$, one obtain the desired state $|{\tilde C}\rag$ whose rank on ${\bar B}$ gives the dimension $d_\tB^{\min}=d_{B^L}$. ${\tilde C}_{{\bar B}{\bar B}_1}$ is separable and contains only classical correlations (dashed line).}
	\label{fig:Cstates}
\end{figure}

\section{Application examples}\label{sec:example}
\subsection{Classical case}
The condition for the classical case is well known. However, here we reproduce it using our results to check for consistency. 
Let us consider the classical bipartite state
\begin{equation*}
    \rho_{AB}=\sum_{a,b}p(a,b)|a,b\rag\lag a,b|
\end{equation*}
such that $\rho_A,\rho_B>0$ without loss of generality. Then, the Choi matrix $J_{AB}$ is 
\begin{equation*}
    J_{AB}=\sum_{a,b}p(a|b)|a,b\rag\lag a,b|
\end{equation*}
and obtain $\omega_{a,b}=p(a|b)$ and $K_{a,b}=\sqrt{p(a|b)}|a\rag\lag b|$. The unital CP map in Eq.~\eqref{def:tOmega} can be written as
\begin{equation*}
 \Omega^\dagger_{A\to B}(\cdot)=\sum_{a,b}p(a|b)\lag a|\cdot |a\rag|b\rag\lag b|\,.
\end{equation*}
Then, we introduce an equivalence relation between basis labels $|b\rag$ as
\begin{equation*}
    b\sim b'\Leftrightarrow p(a|b)=p(a|b')\,,\quad \forall a\,.
\end{equation*}
This equivalence relation induces a disjoint decomposition of the label set $\cB=\{b\}$ as $\cB=\bigsqcup_{i=1}^m\cB^i$ and the Hilbert space as $\cH_B=\bigoplus_i\cH_{B^i}$, where $\cH_{B^i}:={\rm span}\{|b\rag\: |\:b\in\cB^i\}$. We define $I_{B^i}:=\sum_{b\in \cB^i}|b\rag\lag b|$. A simple calculation shows that
\begin{equation*}
    {\rm Im}{\Omega}^\dagger_{A\to B}=\left\{\bigoplus_{i=1}^m c_iI_{B^i} \, \middle\vert \: c_i\in\mathbb{C}\right\}\,,
\end{equation*}
and therefore, 
\begin{equation*}
   \cF(\cT)=\left({\rm Im}\Omega^\dagger_{A\to B}\right)'\cong\bigoplus_{i=1}^m|i\rag\lag i|\otimes M_{B^i}(\mathbb{C})\,
\end{equation*}
from Theorem~\ref{thmET}:
As $\rho_B$ is classical and does not have an off-diagonal term between $b\neq b'$, $\Delta_{\rho_B}^t(\cF(\cT))\subset\cF(\cT)$ is true for any $t\in\mathbb{R}$. Hence, $(\cM_B^S)'=\cF(\cT)$ and the minimal sufficient subalgebra of $\cH_B$ become
\begin{equation*}
   \cM_B^S\cong\bigoplus_{i=1}^m|i\rag\lag i|\otimes I_{B^i}\,
\end{equation*}
that the verification of
\begin{equation*}
   C_{BB_1}=\frac{1}{d_B}\bigoplus_{i=1}^m|i\rag\lag i|_{B_i^L}\otimes |\cI\rrag\llag \cI|_{B_i^RB_{1i}^R}\otimes|i\rag\lag i|_{B_{1i}^L}\,.
\end{equation*}
and 
\begin{equation}
    d_\tB^{\min}=m=\rank(C_{BB_1})\,.
\end{equation}
is simple ($d_{B_i^L}=1$ in this case).
This is consistent with minimal sufficient statistics\cite{Coverthomas06}.

\subsection{Pure bipartite states}
For a pure state, $\rho_{AB}=|\psi\rag\lag\psi|_{AB}$ restricted to $\rho_A,\rho_B>0$ implies $d_A=d_B=\Sch(A)_\psi$. 
\begin{equation}
    J_{AB}=|\cI\rrag\llag \cI|_{AB}
\end{equation}
and thus $\Omega_{A\to B}^\dagger=\id_{A\to B}$. This implies that
\begin{equation}
    \cF(\cT)=\mathbb{C}I_B
\end{equation}
according to Theorem~\ref{thmET}. This is also evident that $\Omega_{B\to E}$ must be a completely depolarized channel. Therefore, $(\cM_B^S)'=\mathbb{C}I$ and $d_\tB^{\min}=\Sch(A:B)_\psi=d_B$; that is, further compression is impossible. 

\subsection{Exact compression of quantum channel}
For a given CPTP-map $\cE_{A\to B}$, $\cF_{B\to\tB}$ is an exact compression if there is a post-processing CPTP-map $\cR_{\tB\to B}$ that satisfies 
\begin{align*}
    {\tilde \cE}_{A\to \tB}&:=\cF_{B\to \tB}\circ\cE_{A\to B}\,,\\
    \cE_{A\to B}&=\cR_{\tB\to B}\circ{\tilde \cE}_{A\to \tB}\,.
\end{align*}
Again, the exact compression is non-trivial if $d_\tB<d_B$. 
Consider the (normalized) Choi state of $\cE_{A\to B}$
\begin{equation}
    \rho_{AB}=(\id_A\otimes\cE_{{\bar A}\to B})(|\Psi\rrag\llag\Psi|_{A{\bar A}})\,.
\end{equation}
Then, the problem is reduced to the exact compression of $\rho_{AB}$ on $B$.

For these Choi states, $J_{AB}$ is further reduced to
\begin{align*}
    J_{AB}&=\cE(\tau_A)^{-1/2}(\id_A\otimes\cE)\left(|\Psi\rrag\llag\Psi|_{A{\bar A}} \right)\cE(\tau_A)^{-1/2}\\
    &=\cE(I_A)^{-1/2}(\id_A\otimes\cE)\left(|\cI\rrag\llag\cI|_{A{\bar A}} \right)\cE(I_A)^{-1/2}\,,
\end{align*}
which is the Choi operator of $(\cR^{\tau,\cE}_{B\to A})^\dagger$; that means, $\Omega_{A\to B}^\dagger=(\cR^{\tau,\cE}_{B\to A})^\dagger$. 
Therefore, $(\cM_B^S)'$ is the largest subalgebra of $({\rm Im}(\cR^{\tau,\cE}_{B\to A})^\dagger)'$ invariant under ${\rm ad}_{\cE_{A\to B}(\tau_A)}$.

\subsubsection{Unital channels}
Consider the case in which $\cE_{A\to B}$ is unital, that is, $\cE(I_A)=I_B$ ($A\cong B$). Then $\Delta_{\rho_B}^t=\id_B$ and $\Omega_{B\to A}=\cR^{\tau,\cE}_{B\to A}=\cE^\dagger_{B\to A}$ are  true. Because the modular invariance is trivial, $(\cM_B^S)'=\cF(\cT)$ and
\begin{equation}\label{eq:unitalMS}
    (\cM_B^S)'=({\rm Im}\cE_{A\to B})'\,,\quad \cM_B^S=\alg\left({\rm Im}\cE_{A\to B}\right)\,.
\end{equation}
from Theorem~\ref{thmET}.

An example of a unital channel is a $G$-twirling operation of a finite group $G$. 
Let $U(g)$ be a unitary representation of $G$ in $\cH_B$. $G$-twirling $\cG(\cdot)$ is defined as follows: 
\begin{equation}
    \cG(\cdot):=\frac{1}{|G|}\sum_{g\in G}U(g)\cdot U(g^{-1})\,.
\end{equation}
In general, $U(g)$ is not irreducible and the Hilbert space $\cH_B$ can be decomposed as 
\begin{equation}
    \cH_B=\bigoplus_{\mu:irrep} \cH_{B^L_\mu}\otimes \cH_{B^R_\mu}\,,
\end{equation}
where the irreducible decomposition of $U(g)$ is given by 
\begin{equation}
    U(g)=\bigoplus_{\mu:irrep} U_\mu(g)\otimes I_{B^R_\mu},
\end{equation}
correspondingly. $U_\mu(g)$ acts on $\cH_{B^L_\mu}$ and $\dim\cH_{B^R_\mu}$ is the multiplicity of $\mu$. Using Schur's lemma and Eq.~\eqref{eq:unitalMS}, we obtain 
\begin{equation}
    d_\tB^{\min}=\sum_{\mu}d_{B_\mu^R}\,
\end{equation}
as desired. 

\section{Conclusion}\label{sec:conclusion}
We studied the task of exact local compression of quantum bipartite states, where the goal is to locally reduce the dimensions of each Hilbert space without losing any information. We demonstrated a closed formula that allows us to calculate the minimum achievable dimensions for this compression task. 

Furthermore, we derived simple upper and lower bounds for the rank of the matrix, which offers a numerically tractable approach to estimate the minimal dimensions. These bounds are in general not tight; however, it is still useful to check, for example, the scaling of the optimal compressed dimension.  Future studies could focus on exploring better (analytical and numerical) expressions for the optimal dimension. It would also be interesting to investigate the relationship between the obtained dimension and other known one-shot entropic quantities. 

The ability to compress quantum states and channels while maintaining their informational content could have more implications for zero-error quantum information processing. Theorem~\ref{thmET} suggests that non-trivial compression is impossible for generic bipartite states or quantum channels. The asymptotic result~\cite{Khanian2019-bg} also relies on the non-trivial Koashi-Imoto decomposition, so considering multiple copies is not helpful. However, it would be possible that for structured states or channels such as tensor network states or covariant channels, exact local compression can be generally non-trivial. The exact local compression on many-body states are also related to renormalization group flow as discussed in Refs.~\cite{mixedphase1,mixedphase2}

Practically, it would also be interesting to investigate {\it approximate} one-shot scenario, which has a better potential for applications in quantum information processing (such as quantum information bottleneck problem~\cite{8849518}) but has not been investigated in this paper. The asymptotic result~\cite{Khanian2019-bg} implies that asymptotic compression with vanishing global error also requires non-trivial Koashi-Imoto decomposition, thus the vanishing error condition may still be too restrictive (see e.g., Ref.~\cite{PhysRevA.107.052421} for a possible modification). One possible way to obtain approximate compression is considering ``smoothed" variant of the minimal dimension obtained in this paper. We leave studying such modification as a future problem.

\section*{Acknowledgments}
The author would like to thank Anna Jen\v{c}ov\'a and Francesco Buscemi for their helpful comments and fruitful discussions. K.~K. acknowledges support from JSPS Grant-in-Aid for Early-Career Scientists, No. 22K13972; from MEXT-JSPS Grant-in-Aid for Transformative Research Areas (A) ``Extreme Universe,” No. 22H05254; from JSPS Grant-in-Aid for Challenging Research (Exploratory), No. 23782782; and from the MEXT-JSPS Grant-in-Aid for Transformative Research Areas (B) No. 24H00829.

\section{Operator-algebra quantum error correction}\label{app:OAQEC}
The theory of operator-algebra quantum error correction introduced in Refs.~\cite{Beny2007-bw,Beny2007-ua} provides the most general framework of quantum error correction in the sense that it is a generalization of both standard quantum error correction~\cite{Bennett_1996,PhysRevA.55.900} and operator quantum error correction~\cite{PhysRevLett.94.180501,kribs2006operatorquantumerrorcorrection}. Standard quantum error correction discusses the condition for a subspace that is left protected from noise (a CPTP-map), and operator quantum error correction discusses that for a noiseless subsystem. Operator-algebra quantum error correction includes both of these frameworks and provides a characterization of the set of observables whose information is perfectly protected from a given noise. 

Formally, the {\it correctable algebra}
$\cA_\cE$ of a CPTP-map $\cE:\cB(\cH)\to\cB(\cK)$ on a code subspace $P_C\cH$ is a subalgebra of $\cB(P_C\cH)$ defined as: 
\begin{align}\label{eq:defA}
\cA_\cE&:=\left\{ X\in \cB(\cH)\;|\;[P_CE_a^\dagger E_bP_C,X]=0\;\forall a,b \right\}\,,
\end{align}
where $\{E_a\}$ are the Kraus operators of $\cE$. $P_C$ is the projection onto $C$, and it is set to be $P_C=I$ in the main text and in the following. Eq.~\eqref{eq:defA} is a generalization of the Knill-Laflamme condition~\cite{PhysRevA.55.900}, and it has been shown that the condition is equivalent to the existence of {\it recovery map} $\cR:\cB(\cK)\to\cB(\cH)$ such that~\cite{Beny2007-bw,Beny2007-ua}
\begin{equation}
    (\cR\circ\cE)^\dagger(X)=X\,,\forall X\in\cA\,.
\end{equation}
This is also equivalent to 
\begin{equation}
    \tr\left(X\rho\right)=\tr\left(X\cR\circ\cE(\rho)\right)\,,\forall \rho\in\cB(\cH)\,,\forall X\in\cA\,,
\end{equation}
i.e., the information of all $X\in\cA$  on any state is perfectly recoverable after the noise is performed. 
The recovery map can be chosen to be the Petz recovery map with respect to the completely mixed state $\tau$, and the following characterization is known.
\setcounter{lem}{2} 
\begin{lem}\cite{Chen2020-in} For any CPTP-map $\cE$, it holds that 
\begin{equation}
    \cA_\cE=\cF\left(\cR^{\tau,\cE}\circ\cE\right)=\cF\left(\cE^\dagger\circ(\cR^{\tau,\cE})^\dagger\right) \,.
\end{equation}
\end{lem}

We emphasize that one should not confuse the minimal sufficient subalgebra and the correctable algebra. Although both algebras have recovery maps, the former characterizes the information needed to keep a given set of states under arbitrary noises, while the latter characterizes the information protected from a given noise. 

In the standard quantum error correction, the no-cloning theorem~\cite{Nielsen} induces a complementary relation: if one can perfectly recover a state from the noise, then one cannot learn any information about the state from the output of the complimentary channel of the noise. In operator-algebra quantum error correction, it is possible to learn some classical information from the complementary channel without violating the no-cloning theorem. 

For a Kraus representation $\{E_a\}$ of $\cE$, a Kraus representation $\{F_a\}$ of the complementary channel $\cE^c:\cB(\cH)\to\cB(\cH_E)$ is given by
\begin{equation}
    F_a:=\sum_b|\phi_b\rag\lag a|E_b\,,
\end{equation}
where $\{|\phi_b\rag\}$ and $\{|a\rag\}$ are orthonormal bases of $\cH_E$ and $\cK$, respectively. The correctable algebra of $\cE^c$ is then generated by operators commuting with
\begin{equation}
    F_a^\dagger F_b=\sum_{c,d}E_c^\dagger|a\rag\lag\phi_c|\phi_d\rag\lag b|E_d=\cE^\dagger(|a\rag\lag b|),,
\end{equation}
thus the following lemma holds.
\setcounter{lem}{1} 
\begin{lem}~\cite{Beny2007-ua}
For any CPTP-map $\cE$, let $\cA_{\cE^c}$ be the correctable algebra of the complementary map $\cE^c$. Then, it holds that
    \begin{equation}
        \cA_{\cE^c}=({\rm Im}\cE^\dagger)'\,.
    \end{equation}
\end{lem}
This is the complementary relation for operator-algebra quantum error correction.

\end{document}